\documentclass[12pt,a4paper]{article}%
\usepackage{amsfonts}
\usepackage{amsthm}
\usepackage{amsmath}
\usepackage{amssymb}
\usepackage{graphicx}%
\newtheorem{theorem}{Theorem}
\newtheorem{lemma}{Lemma}
\newtheorem{remark}{Remark}
\newtheorem{proposition}{Proposition}
\newtheorem{corollary}{Corollary}
\newtheorem{example}{Example}
\begin{document}

\title{\textbf{{\Large {A finite-sum representation for solutions for the Jacobi
operator}} }}
\author{{\large {Hugo M. Campos and Vladislav V. Kravchenko}}\\{\small Departamento de Matem\'{a}ticas, CINVESTAV del IPN, Unidad
Quer\'{e}taro} \\{\small Libramiento Norponiente No. 2000, Fracc. Real de Juriquilla} \\{\small Queretaro, Qro. C.P. 76230 MEXICO } \\{\small e-mail: hcampos@math.cinvestav.mx; vkravchenko@qro.cinvestav.mx}}
\maketitle

\begin{abstract}
We obtain a finite-sum representation for the general solution of the
equation
\[
\Delta\left(  p(n-1)\Delta u(n-1)\right)  +q(n)u(n)=\lambda r(n)u(n)
\]
in terms of a nonvanishing solution corresponding to some fixed value of
$\lambda=\lambda_{0}$. Applications of this representation to some results on
the boundedness of solutions are given as well as illustrating examples.

\end{abstract}


\section{Introduction}

First we introduce some notations. For $a\in\mathbb{R}$ we define\ the
following set $N_{a}=\{a,a+1,...\},$ and $\Delta$ stays for the difference
operator, $\Delta u(n)=u(n+1)-u(n)$. A sequence $v(n)$ satisfying $\Delta
v(n)=u(n)$ is called an indefinite sum of $u(n)$. The indefinite sum of a
sequence is not unique and corresponding indefinite sums differ by a constant.
By $\sum_{j=n_{0}}^{n-1}{}^{\ast}u(j)$ we denote the indefinite sum of $u(n)$
satisfying the boundary condition $u(n_{0})=0$. Then we have \cite{Teschl}:
\[
\sum_{j=n_{0}}^{n-1}{}^{\ast}u(j)=\left\{
\begin{array}
[c]{cl}%
\displaystyle\sum_{j=n_{0}}^{n-1}u(j), & n>n_{0}\\
0, & n=n_{0}\\
-\displaystyle\sum_{j=n}^{n_{0}-1}u(j), & n<n_{0}%
\end{array}
\right.  .
\]

We will consider the second order difference equation of the form
\begin{equation}
\Delta\left(  p(n-1)\Delta u(n-1)\right)  +q(n)u(n)=\lambda r(n)u(n),\hspace
{0cm}n\in N_{a}, \label{SLP}%
\end{equation}
where $r,$ $p,$ $q$ are given complex sequences defined on $N_{a}$, $N_{a-1}$
and $N_{a}$ respectively, $p(n)\neq0$ for all $n\in N_{a-1}$, $\lambda
\in\mathbb{C}$ is the spectral parameter and $u(n)$, defined on $N_{a-1}$, is
the unknown function. The operator on the left-hand side is known as the
Jacobi operator and has been extensively studied (see, e.g., \cite{Teschl}).
If in (\ref{SLP}) $\lambda=0,$ we obtain the equation
\begin{equation}
\Delta\left(  p(n-1)\Delta u(n-1)\right)  +q(n)u(n)=0 \label{SL}%
\end{equation}
which was studied in dozens of works regarding several aspects, namely,
oscillation, disconjugacy, disfocality, asymptotic behaviour, boundedness and
boundary value problem (see, e.g., the books \cite{Agarwal1}, \cite{Agarwal2},
\cite{Elaydi} and \cite{Walter}). Equation (\ref{SLP}) can be regarded as a
discrete analogue of the Sturm-Liouville differential equation
\[
(p(x)y^{\prime}(x))^{\prime}+q(x)y(x)=\lambda r(x)y(x)
\]
and quite often techniques and results developed for (\ref{SLP}) represent
discrete analogues of the corresponding continuous results for the
Sturm-Liouville equation.

In this paper we begin by obtaining a discrete version of some results from
\cite{Kravchenco2} and \cite{Kravchenco1} concerning the spectral power series
representation for the general solution of the Sturm-Liouville differential
equation. This representation being a different form of a perturbation
Liouville-Neumann series \cite{Bellman} offers an efficient algorithm for
numerical calculation of eigenfunctions and eigenvalues of a Sturm-Liouville
problem (see \cite{CKKO}, \cite{Kravchenco2}, \cite{Kravchenco1},
\cite{KhmRosu}).

For linear difference equations a spectral power series representation for
solutions was considered also as a perturbation technique, however even the
situation with the convergence of such series was not satisfactorily
understood (see, e.g., \cite[p. 91]{Agarwal1}, where the possibility of
divergence of the series as those considered in the present work is assumed).
Motivated by \cite{Kravchenco1} and \cite{Kravchenco2}, we propose a different
procedure to find the coefficients of such series for solutions of (\ref{SLP})
(see Theorem \ref{ThRepresentation}) which gives us as a simple corollary that
those series are in fact finite sums (see Lemma \ref{LemmaFinite}).

As an application of this representation, we give alternative proofs of some
results already known in the literature, concerning boundedness of solutions.
We also extend the criterion of the boundedness of all solutions of a linear
second-order difference equation onto a general case of complex coefficients,
Theorem \ref{ThBoundModule}.

\section{A finite-sum representation for solutions}

In this section we prove the main result of the present work, Theorem
\ref{ThRepresentation}, which establishes that any nonvanishing solution
$u_{0}$ of (\ref{SL}) allows us to obtain a general solution of (\ref{SLP}) as
follows. Consider the sequences
\begin{equation}
u_{1}(n)=\left\{
\begin{array}
[c]{ll}%
u_{0}(n)\displaystyle\sum_{k=0}^{n-n_{0}-1}\lambda^{k}X^{(2k)}(n), &
n>n_{0}\vspace{0.3cm}\\
u_{0}(n)\displaystyle\sum_{k=0}^{n_{0}-n}\lambda^{k}X^{(2k)}(n), & n\leq n_{0}%
\end{array}
\right.  \label{s1}%
\end{equation}
and%
\begin{equation}
u_{2}(n)=\left\{
\begin{array}
[c]{ll}%
u_{0}(n)\displaystyle\sum_{k=0}^{|n-n_{0}|-1}\lambda^{k}Y^{(2k+1)}(n), & n\neq
n_{0}\vspace{0.3cm}\\
0, & n=n_{0}\vspace{0.3cm}%
\end{array}
\right.  \label{s2}%
\end{equation}
where $X^{(i)}$ and $Y^{(i)}$ are defined recursively by the relations
\[
X^{(0)}=Y^{(0)}=1,\vspace{0cm}%
\]%
\begin{equation}
X^{(i)}(n)=\left\{
\begin{array}
[c]{ll}%
\displaystyle\sum_{s=n_{0}}^{n-1}{}^{\ast}\dfrac{X^{(i-1)}(s)}{p(s)u_{0}%
(s)u_{0}(s+1)}, & i\,\,\,\,\text{even}\vspace{0.5cm}\\
\displaystyle\sum_{s=n_{0}}^{n-1}{}^{\ast}u_{0}^{2}(s+1)X^{(i-1)}%
(s+1)r(s+1), & i\,\,\,\,\text{odd}%
\end{array}
\right.  \vspace{0cm}\label{xi}%
\end{equation}%
\begin{equation}
Y^{(i)}(n)=\left\{
\begin{array}
[c]{ll}%
\displaystyle\sum_{s=n_{0}}^{n-1}{}^{\ast}u_{0}^{2}(s+1)Y^{(i-1)}%
(s+1)r(s+1), & i\,\,\,\,\text{even}\vspace{0.5cm}\\
\displaystyle\sum_{s=n_{0}}^{n-1}{}^{\ast}\dfrac{Y^{(i-1)}(s)}{p(s)u_{0}%
(s)u_{0}(s+1)}, & i\,\,\,\,\text{odd}%
\end{array}
\right.  \vspace{0.3cm}\label{yi}%
\end{equation}
and $n_{0}\in N_{a-1}$ is an arbitrary point. We show that they are linearly
independent solutions of (\ref{SLP}). In order to prove this statement we need
first the following auxiliary result.

\begin{lemma}
\item \label{LemmaFinite}(i) For $k\geq1$,
\[
m\in\{n_{0}-k+1\,,......\,,n_{0}+k\}\Rightarrow X^{(2k)}(m)=0,
\]
and for $k\geq0$,
\[
m\in\{n_{0},n_{0}\pm1,....,n_{0}\pm k\}\Rightarrow Y^{(2k+1)}(m)=0.
\]
(ii) The sequences defined by the formulas%
\[
u_{1}(n)=u_{0}(n)\displaystyle\sum_{k=0}^{\infty}\lambda^{k}X^{(2k)}%
(n),\hspace{0cm}u_{2}(n)=u_{0}(n)\displaystyle\sum_{k=0}^{\infty}\lambda
^{k}Y^{(2k+1)}(n)
\]
can be written as follows%
\[
u_{1}(n)=\left\{
\begin{array}
[c]{ll}%
u_{0}(n)\displaystyle\sum_{k=0}^{n-n_{0}-1}\lambda^{k}X^{(2k)}(n), &
n>n_{0}\vspace{0.3cm}\\
u_{0}(n)\displaystyle\sum_{k=0}^{n_{0}-n}\lambda^{k}X^{(2k)}(n), & n\leq n_{0}%
\end{array}
\right.
\]%
\[
u_{2}(n)=\left\{
\begin{array}
[c]{ll}%
u_{0}(n)\displaystyle\sum_{k=0}^{|n-n_{0}|-1}\lambda^{k}Y^{(2k+1)}(n), & n\neq
n_{0}\vspace{0.3cm}\\
0, & n=n_{0}\vspace{0.3cm}%
\end{array}
.\right.
\]

\end{lemma}

\begin{proof}
(i) We use the reasoning by induction to prove that if $m\in
\{n_{0}+1,...,n_{0}+k\}$ then
\begin{equation}
X^{(2k)}(m)=0,\qquad k>0.  \label{vsp}
\end{equation}
All other cases contemplated in (i) are treated in a similar way. Note that
by definition $X^{(i)}(n_{0})=0$ for all $i\neq 0$. For $k=1$ relation (\ref%
{vsp}) holds due to the equality
\begin{equation*}
X^{(2)}(n_{0}+1)=\dfrac{X^{(1)}(n_{0})}{p(n_{0})u_{0}(n_{0})u_{0}(n_{0}+1)}%
=0.
\end{equation*}%
Suppose that the assertion is true for $k$. Then by the equality
\begin{equation*}
X^{(2k+1)}(n)=\displaystyle%
\sum_{s=n_{0}}^{n-1}u_{0}^{2}(s+1)X^{(2k)}(s+1)r(s+1),\,\,\,\,\,n>n_{0},
\end{equation*}%
we conclude that $n_{0}+1,...,n_{0}+k$ are zeros of $X^{(2k+1)}$. From this
and due to the relation
\begin{equation*}
X^{(2k+2)}(n)=\displaystyle\sum_{s=n_{0}}^{n-1}\dfrac{X^{(2k+1)}(s)}{%
p(s)u_{0}(s)u_{0}(s+1)},\,\,\,\,n>n_{0}
\end{equation*}%
we obtain that $n_{0}+1,...,n_{0}+k+1$ are zeros of $X^{(2k+2)}$, and the
assertion is valid for $k+1$.
(ii) Let $n>n_{0}$. By part (i), $X^{(2k)}(n)=0$ for all $k\geq n-n_{0},$
thus the series defining $u_{1}$ is actually a sum from $k=0$ to $k=n-n_{0}-1
$. Other cases are proved similarly.
\end{proof}

\begin{theorem}
\label{ThRepresentation}Assume that $u_{0}$ is a nonvanishing solution of
(\ref{SL}). Then the sequences (\ref{s1}) and (\ref{s2}) are linearly
independent solutions of (\ref{SLP}), where $X^{(i)}$ and $Y^{(i)}$ are
defined recursively by the relations (\ref{xi}) and (\ref{yi}) and $n_{0}\in
N_{a-1}$ is an arbitrary point.
\end{theorem}

\begin{proof}
First we prove that the sequences $u_{1}$ and $u_{2}$ defined as follows
\begin{equation}
u_{1}(n)=u_{0}(n)\displaystyle\sum_{k=0}^{\infty }\lambda ^{k}X^{(2k)}(n),%
\hspace{0cm}u_{2}(n)=u_{0}(n)\displaystyle\sum_{k=0}^{\infty }\lambda
^{k}Y^{(2k+1)}(n)  \label{sol}
\end{equation}%
satisfy equation (\ref{SLP}) and are linearly independent. As was shown in
Lemma \ref{LemmaFinite}  these infinite series are in fact the
finite sums (\ref{s1}), (\ref{s2}).  With the help of the nonvanishing solution $u_{0}(n)$
of (\ref{SL}) the Jacobi operator
\begin{equation*}
Lu(n)=\Delta \left( p(n-1)\Delta u(n-1)\right) +q(n)u(n)
\end{equation*}%
can be factorized as follows
\begin{equation*}
Lu(n)=\dfrac{1}{u_{0}(n)}\Delta \left[ p(n-1)u_{0}(n-1)u_{0}(n)\Delta \left(
\dfrac{u(n-1)}{u_{0}(n-1)}\right) \right]
\end{equation*}%
(this is the Polya factorization \cite{Agarwal2}, \cite{Walter}). Applying
the operator $L$ to $u_{1}$ we obtain
\begin{equation*}
\begin{array}{ll}
Lu_{1}(n) & =\dfrac{1}{u_{0}(n)}\Delta \left[ p(n-1)u_{0}(n-1)u_{0}(n)\Delta %
\displaystyle\sum_{k=0}^{\infty }\lambda ^{k}X^{(2k)}(n-1)\right]  \\
& =\dfrac{1}{u_{0}(n)}\Delta \displaystyle\sum_{k=1}^{\infty }\lambda
^{k}X^{(2k-1)}(n-1) \\
& =r(n)u_{0}(n)\displaystyle\sum_{k=1}^{\infty }\lambda
^{k}X^{(2k-2)}(n)=\lambda r(n)u_{1}(n).%
\end{array}%
\end{equation*}%
\qquad
The same technique can be used to prove that $u_{2}(n)$ is a solution as
well. In order to prove that $u_{1}$ and $u_{2}$ are linearly independent,
it is necessary to verify that their Casoratian is different from zero at
any point. Since $u_{2}(n_{0})=0$, the Casoratian of $u_{1}$ and $u_{2}$ at $%
n_{0}$ can be calculated,
\begin{equation*}
W(u_{1},u_{2})(n_{0})=u_{1}(n_{0})u_{2}(n_{0}+1)=\dfrac{u_{0}(n_{0})}{%
p(n)u_{0}(n_{0})}=\dfrac{1}{p(n_{0})}\neq 0.
\end{equation*}
\end{proof}

\begin{example}
Consider the equation $\Delta^{2}u(n-1)=\lambda u(n)$. In this case one can
choose $u_{0}\equiv1$ and $n_{0}=0.$ Then by (\ref{xi}) and (\ref{yi}) we
have
\[
X^{(2k)}(n)=\dfrac{(n+k-1)^{(2k)}}{2k!},\qquad Y^{(2k+1)}(n)=\dfrac
{(n+k)^{(2k+1)}}{(2k+1)!},
\]
where $n^{(k)}:=n(n-1)...(n-k+1)$.
\end{example}

\begin{remark}
Obviously Theorem \ref{ThRepresentation} can also be applied when a
nonvanishing solution of the equation
\[
\Delta\left(  p(n-1)\Delta u(n-1)\right)  +q(n)u(n)=\lambda_{0}r(n)u(n)
\]
is known. In this case the Polya factorization is applied to the operator
$\widetilde{T}=T-\lambda_{0}I$ and (\ref{SLP}) is written in the form
$\widetilde{T}u=(\lambda-\lambda_{0})u$. Then for $n>n_{0}$ the solutions
(\ref{s1}) and (\ref{s2}) are given by the following sums
\[
u_{1}(n)=u_{0}(n)\displaystyle\sum_{k=0}^{n-n_{0}-1}(\lambda-\lambda_{0}%
)^{k}X^{(2k)}(n),\hspace{0.3cm}u_{2}(n)=u_{0}(n)\displaystyle\sum
_{k=0}^{n-n_{0}-1}(\lambda-\lambda_{0})^{k}Y^{(2k+1)}(n),
\]
and for $n\leq n_{0}$ the corresponding representation of solutions is also
obtained from (\ref{s1}) and (\ref{s2}) by replacing $\lambda^{k}$ with
$(\lambda-\lambda_{0})^{k}$.
\end{remark}

\begin{remark}
When $p$ and $q$ are real sequences a nonvanishing solution always exists.
Indeed, two linearly independent real solutions $u$ and $v$ never vanish
simultaneously (because otherwise their Casoratian vanishes) thus one can
choose $u_{0}=u+iv$.
\end{remark}

\begin{example}
Let us consider the equation
\begin{equation}
\Delta\left(  n\Delta u(n-1)\right)  +\lambda u(n)=0,\hspace{0.5cm}n\in N_{1}.
\label{laguerre}%
\end{equation}
Let $n_{0}=0$ and $u_{0}\equiv1$. Using the auxiliary operator
\[
Tu(n):=\sum_{s=0}^{n-1}{}^{\ast}\dfrac{1}{1+s}\sum_{l=0}^{s-1}{}^{\ast
}u(l+1),
\]
from (\ref{xi}) and (\ref{yi}) we get $X^{(2k)}=T(X^{(2k-2)})$ and
$X^{(2k+1)}=T(X^{(2k-1)})$. Then we have the following relations
\[
X^{(2)}(n)=T(1)=\sum_{s=0}^{n-1}{}^{\ast}\frac{s}{1+s}=\sum_{s=0}^{n-1}%
{}^{\ast}1-\sum_{s=0}^{n-1}{}^{\ast}\dfrac{1}{1+s}=n-Y^{(1)}(n),
\]%
\[
X^{(4)}(n)=T(n)-T\left(  Y^{(1)}(n)\right)  =\dfrac{n^{(2)}}{4}-Y^{(3)}(n),
\]%
\[
\vdots
\]%
\begin{equation}
X^{(2k)}(n)=\dfrac{n^{(k)}}{(k!)^{2}}-Y^{(2k-1)}(n). \label{eqx}%
\end{equation}
Consider the following combination of the two solutions
\[
u(n,\lambda):=u_{1}(n)-\lambda u_{2}(n)=(-\lambda)^{n}Y^{(2n-1)}(n)+\sum
_{k=0}^{n-1}(-\lambda)^{k}\dfrac{n^{(k)}}{(k!)^{2}}.
\]
By Lemma \ref{LemmaFinite} $X^{(2n)}(n)=0$, and due to (\ref{eqx}) we have
$Y^{(2n-1)}(n)=\dfrac{n^{(n)}}{(n!)^{2}}$. Thus,
\[
u(n,\lambda)=\sum_{k=0}^{n}(-\lambda)^{k}\dfrac{n^{(k)}}{(k!)^{2}}=\sum
_{k=0}^{n}{\binom{n}{k}}\dfrac{(-\lambda)^{k}}{k!},\hspace{0.5cm}n\geq0.
\]
Note that these are the Laguerre polynomials (of the variable $\lambda$) and
as is well known they satisfy (\ref{laguerre}) (see, e.g., \cite{Elaydi}).
\end{example}

Let us distinguish the following special case of Theorem
\ref{ThRepresentation}.

\begin{corollary}
The sequences $u_{1},$ $u_{2}$ defined by (\ref{s1}), (\ref{s2}), (\ref{xi})
and (\ref{yi}) with $u_{0}\equiv\lambda=1$ are linearly independent solutions
of the equation
\begin{equation}
\Delta\left(  p(n-1)\Delta u(n-1)\right)  =r(n)u(n),\hspace{0.5cm}n\in N_{a}.
\label{SL1}%
\end{equation}

\end{corollary}

\begin{remark}
\label{RemOpT}Let $u_{1}$ and $u_{2}$ be the solutions from the above
corollary. Then for $n>n_{0}$ we have
\begin{equation}
u_{1}(n)=\displaystyle\sum_{k=0}^{n-n_{0}-1}X^{(2k)}(n),\,\,u_{2}%
(n)=\displaystyle\sum_{k=0}^{n-n_{0}-1}Y^{(2k+1)}(n).\label{sol1}%
\end{equation}
Note that using the operator defined by
\begin{equation}
Tu(n):=\displaystyle\sum_{s=n_{0}+1}^{n-1}{}\displaystyle\sum_{\tau=n_{0}%
+1}^{s}\dfrac{u(\tau)r(\tau)}{p(s)},\hspace{0.5cm}n>n_{0}+1,\label{opT}%
\end{equation}
we obtain
\begin{equation}
X^{(2k)}=T(X^{(2k-2)}),\hspace{0.3cm}Y^{(2k+1)}=T(Y^{(2k-1)}).\label{r1}%
\end{equation}

\end{remark}

\section{Applications to results on the boundedness of solutions}

We begin giving another (in our opinion, an easier) proof of an important
result obtained in \cite{Patula}.

\begin{proposition}
\label{PropBound}Let $p(n)>0$ and $r(n)\geq0$. If all solutions of (\ref{SL1})
are bounded then
\[
\displaystyle\sum_{s=a}^{\infty}\displaystyle\sum_{\tau=a}^{s}\dfrac{r(\tau
)}{p(s)}<\infty,\hspace{0.5cm}\displaystyle\sum_{s=a}^{\infty}\dfrac{1}%
{p(s)}<\infty.
\]

\end{proposition}

\begin{proof}
Since the sequences $p$ and $q$ are nonnegative, by definition $X^{(2k)}$ and
$Y^{(2k+1)}$ are nonnegative as well (see Remark \ref{RemOpT}), and as all
solutions of (\ref{SL1}) are bounded, then of course
\begin{equation}
X^{(2)}(n)=T(1)=\displaystyle\sum_{s=n_{0}+1}^{n-1}{}\displaystyle\sum
_{\tau=n_{0}+1}^{s}\dfrac{r(\tau)}{p(s)}\quad\text{and}\quad Y^{(1)}%
(n)=\displaystyle\sum_{s=n_{0}}^{n-1}\dfrac{1}{p(s)}\label{x2}%
\end{equation}
are bounded too.
\end{proof}

The converse of the above result was also proved in \cite{Patula}. The proof
given there works only in the case of nonnegative coefficients. We prove a
more general result.

\begin{theorem}
\label{ThBoundModule}Let $p(n)\neq0$. If
\begin{equation}
\displaystyle\sum_{s=a}^{\infty}\displaystyle\sum_{\tau=a}^{s}\dfrac
{|r(\tau)|}{|p(s)|}<\infty,\hspace{0.5cm}\displaystyle\sum_{s=a}^{\infty
}\dfrac{1}{|p(s)|}<\infty\label{c1}%
\end{equation}
then all solutions of (\ref{SL1}) are bounded.
\end{theorem}

\begin{proof}
Assume that the condition (\ref{c1}) is fulfilled. Then by (\ref{x2}) there exists some $n_{0}$ such that both $|Y^{1}(n)|$ and $%
|X^{(2)}(n)|$ are less than $\delta <1$ for any $n>n_{0}$. By (\ref{opT}),
for $n>n_{0}+1$ we have
\begin{equation*}
|Tu(n)|<\sup \,|X^{(2)}(n)|\cdot \sup |u(n)|.
\end{equation*}%
From this and by (\ref{r1}) we get $|Y^{(2k-1)}(x)|,|X^{(2k)}(n)|<\delta ^{k}
$. Then from (\ref{sol1}) we obtain the boundedness of solutions.
\end{proof}

It is known \cite{Patula} that minimal solutions of (\ref{SL1}) under the
conditions of Proposition \ref{PropBound} and when the condition (\ref{c1}) is
not fulfilled, tend to zero iff there exists such a solution $u$ of
(\ref{SL1}) that the sequence $p(n)\Delta u(n)$ is unbounded. The existense of
such solutions is completely described by the following proposition to which
we also give another and easier proof.

\begin{proposition}
Let $p(n)>0$ and $r(n)\geq0$. Then for every solution $u(n)$ of (\ref{SL1})
the function $\phi(n)=p(n)\Delta u(n)$ is bounded if and only if
\begin{equation}
\displaystyle\sum_{s=a}^{\infty}\displaystyle\sum_{\tau=a}^{s}\dfrac
{r(s+1)}{p(\tau)}<\infty. \label{series bounded}%
\end{equation}

\end{proposition}

\begin{proof}
Denote $\phi _{1,2}(n)=p(n)\Delta u_{1,2}(n)$, where $u_{1,2}$ are given by (%
\ref{sol1}). Then
\begin{equation*}
\phi _{1}(n)=\displaystyle\sum_{k=0}^{n-n_{0}-1}X^{(2k+1)}(n),\hspace{0.3cm}%
\phi _{2}(n)=\displaystyle\sum_{k=0}^{n-n_{0}}Y^{(2k)}(n).
\end{equation*}%
Furthermore, the relations $X^{(2k+1)}=\widetilde{T}(X^{(2k-1)})$ and $%
Y^{(2k+2)}=\widetilde{T}(Y^{(2k)})$ hold, where $\widetilde{T}$ is the
operator defined by
\begin{equation*}
\widetilde{T}u(n):=\displaystyle\sum_{s=n_{0}}^{n-1}\displaystyle\sum_{\tau
=n_{0}}^{s}\dfrac{r(s+1)u(\tau )}{p(\tau )}.
\end{equation*}%
Then the sufficiency of condition (\ref{series bounded}) is proved following
the reasoning from the proof of Proposition \ref{PropBound}, and the
necessity is proved following the reasoning from the proof of Theorem  \ref{ThBoundModule}.
\end{proof}

\end{document}